\documentclass[aos,nameyear,dvips]{arximspdf}
\usepackage{dcolumn}
\usepackage{graphicx}

\doi{10.1214/10-AOS838}
\volume{39}
\issue{1}
\pubyear{2011}
\firstpage{261}
\lastpage{277}

\makeatletter

\newcolumntype{d}[1]{D{.}{.}{#1}}

\newtheorem{theorem}{Theorem}
\newtheorem{prop}{Proposition}
\newtheorem{lemma}{Lemma}

\newcommand{\frn}{\mathfrak{n}}
\newcommand{\frx}{\mathfrak{z}}
\newcommand{\bx}{\mathbf{x}}
\newcommand{\calF}{\mathcal F}
\newcommand{\calN}{\mathcal N}
\newcommand{\rbEstim}[2]{\hat\xi_{#1}^{#2}}

\newcommand{\esp}{\mathbb{E}}
\newcommand{\prob}{\mathbb{P}}
\newcommand{\un}{\mathbb{I}}
\newcommand{\dlim}{\stackrel{\mathcal L}{\longrightarrow}}
\newcommand{\plim}{\stackrel{\mathbb P}{\longrightarrow}}

\makeatother

\begin{document}
\begin{frontmatter}

\title{A vanilla Rao--Blackwellization of Metropolis--Hastings
algorithms\thanksref{T1}}
\runtitle{Vanilla Rao--Blackwellization}

\thankstext{T1}{Supported by the EPSRC and by the Agence Nationale de
la Recherche through the 2009--2012 Project \textsc{Big MC}.}

\begin{aug}
\author[A]{\fnms{Randal} \snm{Douc}\ead[label=e1]{randal.douc@it-sudparis.eu}} and
\author[B]{\fnms{Christian P.} \snm{Robert}\corref{}\ead[label=e2]{xian@ceremade.dauphine.fr}}
\runauthor{R. Douc and C. P. Robert}
\affiliation{Telecom SudParis and Universit\'e Paris-Dauphine}
\address[A]{TELECOM SudParis\\
9 Rue Charles Fourier\\
91011 Evry cedex\\
France\\
\printead{e1}}
\address[B]{CEREMADE\\
Universit\'e Paris-Dauphine\\
75775 Paris cedex 16\\
CREST, INSEE\\
92245 Malakoff cedex\\
France\\
\printead{e2}}
\end{aug}

\received{\smonth{10} \syear{2009}}
\revised{\smonth{5} \syear{2010}}

%
\begin{abstract}
Casella and Robert [\textit{Biometrika} \textbf{83} (1996) 81--94]
presented a general Rao--Black\-wel\-lization principle for
accept-reject and Metropolis--Hastings schemes that leads to
significant decreases in
the variance of the resulting estimators,
but at a high cost in computation and storage. Adopting a completely
different perspective, we introduce instead a
universal scheme that guarantees variance reductions in all
Metropolis--Hastings-based
estimators while keeping the
computation cost under control. We establish a central limit theorem
for the improved estimators and illustrate
their performances on toy examples and on a probit model estimation.
\end{abstract}

%

\end{frontmatter}

\section{Introduction}

As its accept-reject predecessor, the Metropolis--Hastings simulation
algorithm relies
in part on the generation of
uniform variables to achieve given acceptance probabilities. More
precisely, given a target density $f$ with respect
to a dominating measure on the space $\mathcal{X}$, if the
Metropolis--Hastings proposal
is associated with the density
$q(x|y)$ (with respect to the same dominating measure), then the
acceptance probability of the corresponding Metropolis--Hastings iteration
at time $t$ is
\[
\alpha\bigl(x^{(t)},y_t\bigr) = \min\biggl\{ 1, \frac{\pi(y_t)}{\pi
(x^{(t)})}
\frac{q(x^{(t)}|y_t) }{ q(y_t|x^{(t)}) } \biggr\},
\]
where $y_t\sim q(y_t|x^{(t)})$ is the proposed value for $x^{(t+1)}$.
In practice, this means that a uniform $u_t\sim\mathcal{U}(0,1)$ is
first generated and
that $x^{(t+1)}=y_t$ if and only if $u_t\le\alpha(x^{(t)},y_t)$.

Since the uniformity of the $u_t$'s is an extraneous (albeit necessary)
noise, in that it does not directly provide information
about the target $f$ (but only through its acceptance rate),
\citet{casellarobert1996} took advantage of
this flow of auxiliary variables $u_t$ to reduce the variance of the
resulting estimators while preserving
their unbiasedness by integrating out the $u_t$'s conditional on all
simulated $y_t$'s.
This strategy has a nonnegligible cost of $O(N^2)$ for a given sample
of size $N$. While\vadjust{\goodbreak} \mbox{extensions}
have been proposed in the literature [\citet{casellarobert1998},
\citet{perron1999}; see also
\citet{delmasjourdain2009} for an analysis of a
Rao--Blackwellized version of the estimator when
conditioning on the rejected candidates], this solution is therefore
not considered in practice,
in part due to this very cost. The current paper reproduces the
Rao--Blackwellization
argument of \citet{casellarobert1996} by means of an independent
representation that allows the variance to be reduced at a fixed
computational cost. Section \ref{sec:R&B} outlines the
Rao--Blackwellization technique and Section \ref{sec:con} validates
the resulting
variance reduction, including a derivation of the asymptotic variance
of the improved estimators, while Section \ref{sec:ex} presents some
illustrations of the improvement on toy examples.\looseness=-1

\section{The Rao--Blackwellization solution}\label{sec:R&B}
When considering the outcome of a Metropolis--Hastings experiment,
$(x^{(t)})_t$, and
the way it is used in Monte Carlo approximations,
%
%
\begin{equation}\label{eq:estimatorMCMC}
\delta= \frac{1}{N} \sum_{t=1}^N h\bigl(x^{(t)}\bigr) ,
\end{equation}
alternative representations of this estimator are
\[
\delta= \frac{1}{N} \sum_{t=1}^N \sum_{j=1}^t h(y_j) \mathbb
{I}_{x^{(t)}=y_j} \quad\mbox{and}\quad
\delta= \frac{1}{N} \sum_{i=1}^M \frn_i h(\frx_i) ,
\]
where
the $y_j$'s are the proposed Metropolis--Hastings moves,
the $\frx_i$'s are the accepted $y_j$'s,
$M$ is the number of accepted $y_j$'s up to time $N$
and
$\frn_i$ is the number of times $\frx_i$ appears in the sequence
$(x^{(t)})_t$.
The first representation is the one used by
\citet{casellarobert1996}, who integrate out the random elements of the
outer sum,
given the sequence of $y_t$'s. The second representation is also found
in \citet{sahuzhigljavsky1998},
\citet{gasemyr2002}, \citet{sahuzhigljavsky2003} and
\citet{malefakiiliopoulos2008}, and is the basis for our construction.

Let us first recall the basic properties of the pairs $(\frx_i,\frn
_i)$, also found in the above references.
\begin{lemma} \label{lem:induite}
The sequence $(\frx_i,\frn_i)$ is such that: 
%
\begin{enumerate}
\item$(\frx_i,\frn_i)_i$ is a Markov chain; 
%
\item$\frx_{i+1}$ and $\frn_i$ are independent given $\frx_i$;
\item$\frn_i$ is distributed as a geometric random variable with
probability parameter
%
%
\begin{equation}\label{eq:defp}
p(\frx_i) := \int\alpha(\frx_i,y) q(y|\frx_i) \,{d}y ;
\end{equation}
\item$(\frx_i)_i$ is a Markov chain with transition kernel $\tilde
Q(\frx,{d}y)=\tilde q(y|\frx)\,{d}y$ and stationary
distribution $\tilde\pi$ such that 
%
\[
\tilde q(\cdot|\frx) \propto\alpha(\frx,\cdot) q(\cdot|\frx)
\quad\mbox{and}\quad \tilde\pi(\cdot) \propto\pi(\cdot)p(\cdot) .
\]
\end{enumerate}
\end{lemma}
\begin{pf}
We only prove the last point of the lemma. The transition kernel
density $\tilde q$ of the Markov chain\vadjust{\goodbreak} $(\frx_i)_i$ is
obtained by integrating out the geometric waiting time, namely
$
\tilde q(\cdot|\frx_i) = \alpha(\frx_i,\cdot) q(\cdot|\frx_i)
/
p(\frx_i)$.
Thus,
\[
\tilde\pi(x) \tilde q(y|x) = \frac{\pi(x)p(x)}{\int\pi
(u)p(u)\,{d}u}
\frac{\alpha(x,y)q(y|x)}{p(x)}
= \tilde\pi(y) \tilde q(x|y) ,
\]
where we have used the detailed balance property of the original
Metropolis--Hastings algorithm,
namely that $\pi(x) q(y|x)\alpha(x,y)=\pi(y)q(x|y)\alpha(y,x)$. This
shows that the chain
$(\frx_i)_i$ satisfies a detailed balance property with respect to
$\tilde\pi$,
thus that it is $\tilde\pi$-reversible, which completes the proof.
\end{pf}

Since the Metropolis--Hastings estimator $\delta$ only involves the
$\frx_i$'s, that
is, the accepted $y_t$'s,
an optimal weight for those random variables is the importance weight
$1/p(\frx_i)$,
leading to the corresponding importance sampling estimator,
\[
\delta^* = \frac{1}{N} \sum_{i=1}^M \frac{h(\frx_i)}{p(\frx
_i)} ,
\]
but this quantity is usually unavailable in closed form and needs to be
estimated by an unbiased estimator. The geometric $\frn_i$ is the
obvious solution that is used
in the original Metropolis--Hastings estimate, but solutions with
smaller variance also
are available, as shown by
the following results.
\begin{lemma}\label{lem:2uston}
If $(y_j)_j$ is an i.i.d. sequence with distribution $q(y|\frx_i)$,
then the quantity
\[
\hat\xi_i = 1+\sum_{j=1}^\infty\prod_{\ell\le j} \{ 1 -
\alpha
(\frx_i,y_\ell) \}
\]
is an unbiased estimator of $1/p(\frx_i)$, the variance of which,
conditional on $\frx_i$,
is lower than the conditional variance of $\frn_i$, $\{1-p(\frx_i)\}
/p^2(\frx_i)$.
\end{lemma}
\begin{pf} 
Since $\frn_i$ can be written as
\[
\frn_i =1+ \sum_{j=1}^\infty\prod_{\ell\le j} \mathbb{I}\{
u_\ell
\ge\alpha(\frx_i,y_\ell) \} ,
\]
where the $u_j$'s are i.i.d. $\mathcal{U}(0,1)$, given that the sum
actually stops with the first pair $(u_j,y_j)$
such that $u_j \le\alpha(\frx_i,y_j)$, a Rao--Blackwellized version of
$\frn_i$ consists in its expectation conditional
on the sequence $(y_j)_j$:
\begin{eqnarray*}
\hat\xi_i &=& 1+\sum_{j=1}^\infty
\mathbb{E}\biggl[\prod_{\ell\le j}\mathbb{I}\{
u_\ell\ge\alpha(\frx_i,y_\ell)\}\Big| (y_t)_{t\ge
1}\biggr]\\
&=& 1+ \sum_{j=1}^\infty\prod_{\ell\le j} \mathbb{P}\bigl(
u_\ell\ge\alpha(\frx_i,y_\ell)| (y_t)_{t\ge1}\bigr)\\
&=& 1+ \sum_{j=1}^\infty\prod_{\ell\le j} \{ 1 - \alpha(\frx
_i,y_\ell) \} .
\end{eqnarray*}
Therefore, since $\hat\xi_i$ is a conditional expectation of $\frn_i$,
its variance is necessarily smaller.
\end{pf}

We note that this unbiased estimate of $1/p(\frx_i)$ can be related to
the Bernoulli factory approach of
\citet{latuszynskikosmidispaparoberts2010}, in that we are
only using Bernoulli events in this derivation.

Given that $\alpha(\frx_i,y_j)$ involves a ratio of probability
densities, $\alpha(\frx_i,y_j)$ takes the
value $1$ with positive probability and the sum $\hat\xi_i$ is
therefore almost
surely finite. This may, however, require far too many iterations to be
realistically computed or
it may involve too much variability in the number of iterations thus required.
An intermediate estimator with a fixed computational cost is
fortunately available.
\begin{prop} \label{prop:var}
If $(y_j)_j$ is an i.i.d. sequence with distribution $q(y|\frx_i)$ and
$(u_j)_j$ is an i.i.d. uniform sequence,
for any $k\ge0$, the quantity
%
%
\begin{equation}\label{eq:defRBestim}
\hat\xi_i^k = 1+\sum_{j=1}^\infty
\prod_{1\le\ell\le k\wedge j} \{ 1 - \alpha(\frx_i,y_j)
\}
\prod_{k+1\le\ell\le j} \mathbb{I}\{u_\ell\ge\alpha(\frx
_i,y_\ell
)\}
\end{equation}
is an unbiased estimator of $1/p(\frx_i)$ with an almost sure finite
number of terms.
Moreover, for $k \geq1$,
\[
\mathbb{V} [ \hat\xi_i^k|\frx_i]=\frac
{1-p(\frx_i)}{p^2(\frx_i)} -
\frac
{1-(1-2p(\frx_i)+r(\frx_i))^k}{2p(\frx_i)-r(\frx_i)}\biggl( \frac
{2-p(\frx_i)}{p^2(\frx_i)} \biggr)\bigl(p(\frx_i)-r(\frx_i)\bigr) ,
\]
where $p$ is defined in (\ref{eq:defp}) and $r(\frx_i) := \int\alpha
^2(\frx_i,y) q(y|\frx_i) \,{d}y$.
Therefore, we have
\[
\mathbb{V} [ \hat\xi_i|\frx_i] \leq
\mathbb{V} [ \hat\xi_i^k|\frx_i]
\leq
\mathbb{V} [ \hat\xi_i^0|\frx_i] =
\mathbb{V} [ \frn_i|\frx_i] .
\]
\end{prop}

The truncation at the $k$th proposal thus allows for a calibration of
the computational effort since
$\hat\xi_i^k$ costs on average $k$ additional simulations of $y_j$ and
computations of $\alpha(\frx_i,y_j)$
to compute $\hat\xi_i^k$, when compared with the regular
Metropolis--Hastings weight
$\frn_i$.
\begin{pf*}{Proof of Proposition \ref{prop:var}}
Define $y=(y_j)_{j\geq1}$ and $u_{k\dvtx\infty}=(u_\ell)_{\ell\geq k}$.
Note that $\hat\xi_i^0=\frn_i$ and therefore
the conditional variance of $\hat\xi_i^0$ is the variance of a
geometric variable. Now, obviously,
$\rbEstim{i}{k+1}=\mathbb{E} [ \rbEstim{i}{k}
|\frx_i,y,u_{k+2\dvtx\infty}]$;
thus, we have
\[
\mathbb{V} [ \rbEstim{i}{k}|\frx_i
]=\mathbb{V} [ \rbEstim{i}{k+1}|\frx_i]+
\mathbb{E} [ \mathbb{V} [ \rbEstim
{i}{k}|\frx_i,y, u_{k+2\dvtx\infty}]|\frx_i
] .
\]
To get a closed-form expression for the second term on the right-hand
side, we first introduce a geometric random variable $T_k$ defined by
\[
T_k =1+ \sum_{j=1}^\infty\prod_{\ell\le j} \mathbb{I}\{
u_{k+\ell
} \ge\alpha(\frx_i,y_{k+\ell})\} .
\]
Then, by straightforward algebra, $\rbEstim{i}{k}$ may be rewritten as
\[
\rbEstim{i}{k}=C+ \Biggl(\prod_{\ell=1}^k \{ 1 - \alpha(\frx_i,y_j)
\} \Biggr)T_{k+2}
\un\{u_{k+1}>\alpha(\frx_i,y_{k+1})\},
\]
where $C$ does not depend on $u_1, \ldots, u_{k+1}$. Thus,
\[
\mathbb{V} [ \rbEstim{i}{k}|\frx_i, y,
u_{k+2\dvtx\infty}]=
\Biggl(\prod_{\ell=1}^k \{ 1 - \alpha(\frx_i,y_j)
\}^2 \Biggr) T_{k+2}^2 \alpha\{ \frx
_i,y_{k+1})\bigl(1-\alpha(\frx
_i,y_{k+1})\bigr).
\]
Taking the expectation of the above expression, we obtain
\[
\esp(\mathbb{V} [ \rbEstim{i}{k}|\frx
_i,y, u_{k+2\dvtx\infty}])
=\bigl(1-2p(\frx_i)+r(\frx_i)\bigr)^k\biggl( \frac{2-p(\frx_i)}{p^2(\frx_i)}
\biggr)\bigl(p(\frx_i)-r(\frx_i)\bigr) ,
\]
which completes the proof.
\end{pf*}

\section{Convergence properties}\label{sec:con}

Using those Rao--Blackwellized versions of $\delta$ brings about an
asymptotic improvement
for the estimation of $\mathbb{E}_\pi[h(X)]$, as shown by the following
result which, for any $M>0$, compares the estimators $(k \geq0)$
\[
\delta_M^k = \frac{\sum_{i=1}^M \hat\xi_i^k h(\frx_i)}{\sum
_{i=1}^M \hat
\xi_i^k} .
\]
For any positive function $\varphi$, we denote by ${\mathcal
C}_{\varphi}=\{
h;
|h/\varphi|_\infty<\infty\}$
the set of functions bounded by $\varphi$ up to a constant and we assume
that the reference importance sampling estimator is sufficiently well behaved,
in that there exist positive functions $\varphi\geq1$ and $\psi$
such that
%
%
\begin{eqnarray}
\label{eq:lgnInduite}
\forall h \in{\mathcal C}_{\varphi}\qquad \frac{\sum_{i=1}^M
h(\frx
_i)/p(\frx
_i)}{\sum_{i=1}^M 1/p(\frx_i)}&\plim&\pi(h),\\
\label{eq:cltInduite}
\forall h \in{\mathcal C}_{\psi}\qquad \sqrt{M}\biggl( \frac{\sum_{i=1}^M
h(\frx_i)/p(\frx_i)}{\sum_{i=1}^M 1/p(\frx_i)}-\pi(h)\biggr)
&\dlim&\calN
(0,\Gamma(h)).
\end{eqnarray}
\begin{theorem}\label{thm:primo}
Under the assumption that $\pi(p)>0$, the following convergence
properties hold:
\begin{longlist}
\item if $h$ is in ${\mathcal C}_{\varphi}$, then
\[
\delta_M^k \mathop{\plim}_{M \to\infty} \pi(h);
\]
\item if, in addition, $h^2/p \in{\mathcal C}_{\varphi}$ and
$h\in
{\mathcal C}_{\psi}$, then
%
%
\begin{equation} \label{eq:CLTk}
\sqrt{M}\bigl(\delta_M^k - \pi(h)\bigr) \mathop{\dlim}_{M \to\infty} \calN\bigl(0, V_k[h
-\pi(h)]\bigr) ,
\end{equation}
where $V_k(h):=\pi(p) \int\pi({d}\frx) \mathbb{V}
[ \hat\xi_i^k|\frx]h^2(\frx)p(\frx) +\Gamma(h)$.
\end{longlist}
\end{theorem}
\begin{pf}
We will prove that for all $g \in{\mathcal C}_{\varphi}$,
%
%
\begin{equation}
\label{eq:lgn_xik}
M^{-1}\sum_{i=1}^M \hat\xi_i^k g(\frx_i) \plim\pi(g)/\pi(p) .
\end{equation}
Then,\vspace*{1pt} (i) directly follows from (\ref{eq:lgn_xik}) applied to both
$g=h$ and $g=1$.
Now, denote by $\calF_i$ the $\sigma$-field $\calF_i:=\sigma(\frx
_{1},\ldots,\frx_{i+1},\hat\xi_1^k,\ldots,\hat\xi_i^k)$. Since
$\mathbb{E} [ \hat\xi_i^k g(\frx_i)|\calF
_{i-1}]=g(\frx_i)/p(\frx
_i)$, we have
\[
M^{-1}\sum_{i=1}^M \hat\xi_i^k g(\frx_i)=
\Biggl( \sum_{i=1}^M U_{M,i} - \mathbb{E} [
U_{M,i}|\calF_{i-1}]\Biggr)
+ M^{-1}\sum_{i=1}^M g(\frx_i)/p(\frx_i)
\]
with $U_{M,i}:=M^{-1}\hat\xi_i^k g(\frx_i)$. First, consider the second
term on the right-hand side.
Since $\varphi\geq1$, the function $p$ is in ${\mathcal C}_{\varphi}$;
equation (\ref{eq:lgnInduite}) then implies that
$M/\{\sum_{i=1}^M 1/p(\frx_i)\} \plim\pi(p)>0$ and therefore that
%
%
\begin{equation}\label{eq:lgnUnNormalized_p}
\forall g \in{\mathcal C}_{\varphi}\qquad M^{-1}\sum_{i=1}^M g(\frx
_i)/p(\frx
_i) \plim\pi(g)/\pi(p) .
\end{equation}
It remains to check that $\sum_{i=1}^M U_{M,i} - \mathbb{E}
[ U_{M,i}|\calF_{i-1}] \plim0$.
We use asymptotic results for conditional triangular arrays of random
variables given in
Douc and Moulines (\citeyear{doucmoulines2008}), Theorem 11.
Obviously, since $|g|\in{\mathcal C}_{\varphi}$, we have
\[
\sum_{i=1}^M \mathbb{E} [ |U_{M,i}||\calF
_{i-1}] = M^{-1}\sum_{i=1}^M
|g(\frx_i)|/p(\frx_i) \plim\pi(|g|)/\pi(p)
\]
and we only need to show that $\sum_{i=1}^M \mathbb{E} [
|U_{M,i}|\un\{ |U_{M,i}|>\varepsilon\}|\calF_{i-1}] \plim0$.
Let $C>0$ and note that $\{|U_{M,i}|>\varepsilon\} \subset\{|g(\frx
_i)|>(\varepsilon M)/C\} \cup\{\hat\xi_i^k >C\}$.
Again, using $\mathbb{E} [ \hat\xi_i^k g(\frx_i)
|\calF_{i-1}]=g(\frx
_i)/p(\frx_i)$, we have
%
%
\begin{eqnarray} \label{eq:tensionLGN}
&&\sum_{i=1}^M \mathbb{E} [ |U_{M,i}|\un\{
|U_{M,i}|>\varepsilon\}|\calF_{i-1}]\nonumber\\[-8pt]\\[-8pt]
&&\qquad\leq
\frac{1}{M} \sum_{i=1}^M \frac{|g(\frx_i)|\un\{|g(\frx
_i)|>(\varepsilon
M)/C\}}{p(\frx_i)}
+\frac{1}{M} \sum_{i=1}^M \frac{F_C(\frx_i)}{p(\frx_i)}
\nonumber
\end{eqnarray}
with $F_C(\frx_i):=|g(\frx_i)|\mathbb{E} [ \hat\xi
_i^k \un\{\xi_i^k>C\} |\frx_i] p(\frx_i)$. Since
$F_C\leq|g|$,
we have $F_C \in{\mathcal C}_{\varphi}$. Then, again using (\ref
{eq:lgnUnNormalized_p}),
we have
\begin{eqnarray*}
\frac{1}{M} \sum_{i=1}^M \frac{|g(\frx_i)|\un\{|g(\frx
_i)|>(\varepsilon
M)/C\}}{p(\frx_i)} &\plim& 0 , \\
\frac{1}{M} \sum_{i=1}^M \frac{F_C(\frx_i)}{p(\frx_i)} &\plim& \pi
(F_C)/\pi(p) ,
\end{eqnarray*}
which can be arbitrarily small when taking $C$ sufficiently large.
Indeed, using
Lebesgue's theorem in the definition of $F_C$, for any fixed $\frx$,
$\lim_{C \to\infty}F_C(\frx) =0$
and then, again using Lebesgue's theorem, $\lim_{C \to\infty}\pi
(F_C) =0$.
Finally, (\ref{eq:lgn_xik}) is proved. The proof of (i) follows.

We now consider (ii). Without loss of generality, we assume that
$\pi(h)=0$. Write
\[
\sqrt{M}\delta_M^k =\frac{M^{-1/2}\sum_{i=1}^M \hat\xi_i^k h(\frx
_i)}{M^{-1}\sum_{i=1}^M \hat\xi_i^k} .
\]
By (\ref{eq:lgn_xik}), the denominator of the right-hand side converges
in probability to $1/\pi(p)$. Thus,
by Slutsky's lemma, we only need to prove a central limit theorem for
the numerator of the right-hand side.
Define $U_{M,i}:=M^{-1/2}\hat\xi_i^k h(\frx_i)$ and write
\[
M^{-1/2}\sum_{i=1}^M \hat\xi_i^k h(\frx_i)=\Biggl( \sum_{i=1}^M U_{M,i}
- \mathbb{E} [ U_{M,i}|\calF_{i-1}]\Biggr)
+M^{-1/2} \sum_{i=1}^M h(\frx_i)/p(\frx_i) .
\]
Since $h \in{\mathcal C}_{\psi}$ and $M^{-1}\sum_{i=1}^M 1/p(\frx
_i) \plim
1/\pi(p)$, the second term, thanks again to Slutsky's lemma and
equation (\ref{eq:cltInduite}), converges in distribution to $\calN
(0,\allowbreak\Gamma(h)/\pi^2(p))$. Now, consider
the first term on the right-hand side. We will once again use
asymptotic results on triangular arrays of random variables
[as in Douc and Moulines (\citeyear{doucmoulines2008}), Theorem 13].
We have
\begin{eqnarray*}
&&\sum_{i=1}^M \mathbb{E} [ U_{M,i}^2|\calF
_{i-1}]-(\mathbb{E} [ U_{M,i}|\calF
_{i-1}])^2 \\
&&\qquad=M^{-1}\sum_{i=1}^M (h^2(\frx_i) \mathbb{V} [
\hat\xi_i^k|\frx_i]p(\frx_i))/p(\frx_i)\\
&&\qquad\plim
\pi[ \mathbb{V} [ \hat\xi_i^k|\cdot
]h^2(\cdot)p(\cdot)]/ \pi(p) ,
\end{eqnarray*}
by (\ref{eq:lgnUnNormalized_p}) applied to the nonnegative function
$\frx_i \mapsto h^2(\frx_i) \mathbb{V} [ \hat\xi
_i^k|\frx_i]p(\frx_i)$
which is in ${\mathcal C}_{\varphi}$ since it is bounded from above by $h^2/p
\in{\mathcal C}_{\varphi}$.
It remains to show that, for any $\varepsilon>0$,
%
%
\begin{equation}\label{eq:tensionCLT}
\sum_{i=1}^M \mathbb{E} \bigl[ |U_{M,i}|^2\un
_{|U_{M,i}|>\varepsilon}|\calF_{i-1}\bigr]
\plim0 .
\end{equation}
Following the same lines as in the proof of (i), note that for any $C>0$,
we have $\{|U_{M,i}|>\varepsilon\} \subset\{|h(\frx_i)|>(\varepsilon
\sqrt
{M})/C\} \cup\{\hat\xi_i^k >C\}$. Using the fact that
\[
\mathbb{E} [ (\hat\xi_i^k)^2 |\calF
_{i-1}]= \mathbb{V} [ \hat\xi_i^k|\frx
_i]+(\mathbb{E} [ \hat\xi_i^k|\frx
_i])^2\leq
2/p^2(\frx_i),
\]
we have
\begin{eqnarray*}
&&\sum_{i=1}^M \mathbb{E} [ |U_{M,i}|\un\{
|U_{M,i}|>\varepsilon\}|\calF_{i-1}]\\
&&\qquad\leq\frac{2}{M} \sum_{i=1}^M \frac{h^2(\frx_i)\un\{|h(\frx
_i)|>(\varepsilon\sqrt{M})/C\}}{p^2(\frx_i)}
+\frac{1}{M} \sum_{i=1}^M \frac{F_C(\frx_i)}{p(\frx_i)}
\end{eqnarray*}
with $F_C(\frx_i):=h^2(\frx_i)\mathbb{E} [ (\hat
\xi_i^k )^2 \un\{ \xi_i^k>C\}|\frx_i] p(\frx
_i)$. Since $F_C\leq(2h^2)/p$ and $h^2/p
\in{\mathcal C}_{\varphi}$, we have $F_C \in{\mathcal C}_{\varphi
}$. Then, again
using (\ref{eq:lgnUnNormalized_p}),
\begin{eqnarray*}
\frac{1}{M} \sum_{i=1}^M \frac{(h^2(\frx_i)/p(\frx_i))\un\{
|h(\frx
_i)|>(\varepsilon\sqrt{M})/C\}}{p(\frx_i)}&\plim&0 , \\
\frac{1}{M} \sum_{i=1}^M \frac{F_C(\frx_i)}{p(\frx_i)} &\plim& \pi
(F_C)/\pi(p) ,
\end{eqnarray*}
which can be made arbitrarily small by taking $C$ sufficiently large.
Indeed, as in the proof of (i), one can use Lebesgue's theorem in
the definition of
$F_C$ so that for any fixed $\frx$, $\lim_{C \to\infty}F_C(\frx) =0$.
Then, again using Lebesgue's theorem,
$\lim_{C \to\infty}\pi(F_C) =0$. Finally, (\ref{eq:tensionCLT}) is
proved. The proof of (ii) follows.
\end{pf}

The main consequence of this central limit theorem is thus that,
asymptotically, the correlation between the $\xi_i$'s vanishes, hence that
the variance ordering on the $\xi_i$'s extends to the same ordering on
the $\delta_M$'s.

It remains to link the central limit theorem of the usual Markov chain
Monte Carlo (MCMC) estimator (\ref{eq:estimatorMCMC}) with the
central limit theorem expressed in (\ref{eq:CLTk}), with $k=0$
associated with the accepted values.
We will need some additional assumptions, starting with a maximal
inequality for the Markov chain
$(\frx_i)_{i }$: there exists a measurable function $\zeta$ such that
for any starting point $x$, 
%
%
\begin{equation}\label{eq:defMaxim}
\forall h \in{\mathcal C}_{\zeta}\qquad \prob_x\Biggl(\Biggl|\sup
_{0\leq i
\leq
N} \sum_{j=0}^i [ h(\frx_j)-\tilde\pi(h)]
\Biggr|>\varepsilon
\Biggr) \leq\frac{ N C_h(x)}{\varepsilon^2},
\end{equation}
where $\prob_x$ is the probability measure induced by the Markov chain
$(\frx_i)_{i \geq0}$ starting from $\frx_0=x$.

Moreover, we assume that there exists a measurable function $\phi\geq
1$ such that for any starting point $x$,
%
%
\begin{equation} \label{eq:convTildeQ}
\forall h \in{\mathcal C}_{\phi}\qquad \tilde Q^n(x,h)\plim\tilde
\pi
(h)=\pi
(ph)/\pi(p) ,
\end{equation}
where $\tilde Q$ is the transition kernel of $(\frx_i)_{i }$ expressed
in Lemma \ref{lem:induite}.
%
\begin{theorem} \label{thm:lienMCMCInduite} In addition to the
assumptions of Theorem \ref{thm:primo},
assume that $h$ is a measurable function such that
$h/p \in{\mathcal C}_{\zeta}$ and $\{C_{h/p},h^2/p^2\} \subset
{\mathcal C}_{\phi}$.
Assume, moreover, that 
%
\[
\sqrt{M}\bigl( \delta_M^0-\pi(h)\bigr) \dlim\calN\bigl(0,V_0[h-\pi
(h)]\bigr) .
\]
Then, for any starting point $x$,
\[
\sqrt{M_N}\biggl( \frac{\sum_{t=1}^N h(x^{(t)})}{N}-\pi(h)\biggr)
\mathop{\dlim}
_{N \to\infty} \calN\bigl(0,V_0[h-\pi(h)]\bigr) ,
\]
where $M_N$ is defined by
%
%
\begin{equation}
\label{eq:defMn}
\sum_{i=1}^{M_N} \hat\xi_i^0 \leq N < \sum_{i=1}^{M_N+1} \hat\xi
_i^0 .
\end{equation}
\end{theorem}
\begin{pf}
Without loss of generality, we assume that $\pi(h)=0$. In this proof,
we will denote by $\prob_x$ (resp., $\esp_x$)
the probability\vspace*{1pt} (resp., expectation) associated with the Markov chain
$(x^{(t)})_{t \geq0}$ starting from a fixed point $x$.
Using (\ref{eq:lgn_xik}) with $g=1$, one may divide (\ref{eq:defMn}) by
$M_N$ and let $N$ go to infinity. This yields
that $M_N/N \plim\pi(p)>0$. Then, by Slutsky's lemma, Theorem \ref
{thm:lienMCMCInduite} will be proven if we are able to show that
\[
\sqrt{N}\biggl( \frac{\sum_{t=1}^N h(x^{(t)})}{N}-\pi(h)\biggr)
\mathop{\dlim}_{N
\to\infty} \calN\bigl(0,V_0[h-\pi(h)]/\pi(p)\bigr) .
\]
To that end, consider the decomposition
\[
N^{-1/2}\sum_{t=1}^N h\bigl(x^{(t)}\bigr):=\Delta_{N,1}+\Delta_{N,2}+\Delta
_{N,3} ,
\]
where $M_N^\star:=\lfloor N \pi(p)\rfloor$,
\begin{eqnarray*}
\Delta_{N,1}&:=&N^{-1/2} \Biggl(N - \sum_{i=1}^{M_N} \hat\xi_i^0
\Biggr)h(\frx_{M_N+1}) ,\\
\Delta_{N,2}&:=&N^{-1/2}\Biggl(\sum_{i=1}^{M_N} \hat\xi_i^0 h(\frx_i)-
\sum_{i=1}^{M_N^\star} \hat\xi_i^0 h(\frx_i)\Biggr) ,\\
\Delta_{N,3}&:=&N^{-1/2}\sum_{i=1}^{M_N^\star} \hat\xi_i^0 h(\frx_i) .
\end{eqnarray*}
Using the fact that $0\leq N - \sum_{i=1}^{M_N}\hat\xi_i^0\leq\hat
\xi
_{M_N+1}^0$ and Markov's inequality, we have
\[
\prob_x(|\Delta_{N,1}|>\varepsilon) \leq\frac{\esp_x(\hat\xi_{M_N+1}^0
|h(\frx_{M_N+1})|)}{\varepsilon\sqrt{N}}=\frac{\tilde
Q^{M_N+1}(x,|h|/p)}{\varepsilon\sqrt{N}},
\]
which converges in probability to 0 using the facts that $|h|/p \leq
h^2/p^2+1$ and $\{h^2/p^2, 1\} \subset{\mathcal C}_{\phi}$.
Thus, $\Delta_{N,1}\plim0$. We now consider $\Delta_{N,2}$. Note that
%
%
\begin{equation} \label{eq:majoDelta2Maxim}
\prob_x(|\Delta_{N,2}|>\varepsilon) \leq\prob_x\bigl(|A_N|>\varepsilon\sqrt
{N}/2\bigr)+ \prob_x\bigl(|B_N|>\varepsilon\sqrt{N}/2\bigr)
\end{equation}
with
\[
A_N=\sum_{i=M_N\wedge M_N^\star}^{M_N\vee M_N^\star} h(\frx
_i)/p(\frx
_i) \quad\mbox{and}\quad
B_N=\sum_{i=M_N\wedge M_N^\star}^{M_N\vee M_N^\star} \bigl( \hat\xi
_{j}^0-1/p(\frx_i)\bigr)h(\frx_i) .
\]
Now, pick an arbitrary $\alpha\in(0,1)$ and set $\underline
{M}_N:=M_N^\star(1-\alpha)$
and $\overline{M}_N:=M_N^\star(1+\alpha)$. Since $M_N/N \plim\pi(p)$
for all\vspace*{1pt} $\eta>0$, there exists $N_0$
such that for all $N \geq N_0$, $\prob_x(\underline{M}_N \leq M_N
\leq
\overline{M}_N)\geq1-\eta$. Then,
obviously for $N \geq N_0$, the first term on the right-hand side of
(\ref{eq:majoDelta2Maxim}) is bounded by
%
%
\begin{eqnarray}\label{eq:majoAN}
&&\prob_x\bigl(|A_N|>\varepsilon\sqrt{N}/2\bigr) \nonumber\\
&&\qquad\leq \eta+ \prob_x\Biggl(\sup
_{M_N^{\star}
\leq i\leq\overline{M}_N}\Biggl| \sum_{j=M_N^\star}^i h(\frx
_{j})/p(\frx
_j)\Biggr|
>\varepsilon\sqrt{N}/2\Biggr)\\
&&\qquad\quad{} + \prob_x\Biggl(\sup_{\underline{M}_N^{\star} \leq i\leq M_N^\star
}\Biggl|
\sum_{j=i}^{M_N^\star} h(\frx_{j})/p(\frx_j)\Biggr| >\varepsilon
\sqrt
{N}/2\Biggr) .\nonumber
\end{eqnarray}
Using (\ref{eq:defMaxim}), the second term of the right-hand side is
bounded by
\[
4 \overline{M}_N -M_N^\star\esp_x[C_{h/p}(\frx_{M_N^\star
})]/{\varepsilon
^2 N} ,
\]
which converges to $4\alpha\pi(p) \tilde\pi(C_{h/p})/\varepsilon^2$
as $N$
goes to infinity, using the fact that $C_{h/p} \in{\mathcal C}_{\phi
}$. The
resulting bound can thus be arbitrarily small as $\alpha$ goes to 0.
Similarly, one can bound the third term on the right-hand side of
(\ref{eq:majoAN}) and let $N$ go to infinity. Again letting $\alpha$
go to 0, we obtain that $A_N/\sqrt{N} \plim0$. Similarly, the second
term of the right-hand side of (\ref{eq:majoDelta2Maxim}) is bounded by
%
%
\begin{eqnarray}\label{eq:majoBN}
&&\prob_x\bigl(|B_N|>\varepsilon\sqrt{N}/2\bigr) \nonumber\\
&&\qquad\leq\eta+ \prob_x\Biggl(\sup
_{M_N^{\star} \leq i\leq\overline{M}_N}\Biggl| \sum_{j=M_N^\star}^i
\biggl(\hat\xi_{j}^0-\frac{1}{p(\frx_j)}\biggr)h(\frx_{j})\Biggr|
>\varepsilon\sqrt{N}/2\Biggr)\\
&&\qquad\quad{}+ \prob_x\Biggl(\sup_{\underline{M}_N^{\star} \leq i\leq M_N^\star}
\Biggl| \sum_{j=i}^{M_N^\star} \biggl(\hat\xi_{j}^0-\frac{1}{p(\frx
_j)}\biggr) h(\frx_{j})\Biggr| >\varepsilon\sqrt{N}/2\Biggr).\nonumber
\end{eqnarray}
We write $R_N=\sum_{\ell=1}^N ( \hat\xi_{\ell}^0-\frac
{1}{p(\frx
_\ell)})h(\frx_{\ell})$.
Clearly, $(R_N)$ is a $\calF$-martingale where $\calF=(\calF_{i})_{i
\geq1}$ and $\calF_i$ is
the $\sigma$-field $\calF_i:=\sigma(\frx_{1},\ldots,\frx
_{i+1},\hat\xi
_1^0,\ldots,\hat\xi_i^0)$. Then,
by Kolmogorov's inequality, one can bound the second term of (\ref
{eq:majoBN}) in the following way:
\begin{eqnarray*}
&&
\prob_x\Bigl({\sup_{M_N^{\star} \leq i\leq\overline{M}_N}}|
R_i-R_{M_N}| >\varepsilon
\sqrt{N}/2\Bigr) \\
&&\qquad\leq 4 \frac{\esp_x[(R_{M_N^{\star
}}-R_{M_N})^2]}{\varepsilon^2 N }
=\frac{4}{\varepsilon^2 N} \esp_x\Biggl[\sum_{i=M_N^{\star
}}^{\overline
{M}_N}\frac{1-p(\frx_i)}{p^2(\frx_i)}h^2(\frx_i)\Biggr]
\\
&&\qquad=\frac{4(\overline{M}_N-M_N^{\star}+1)}{\varepsilon^2 N} \frac{\sum
_{i=M_N^{\star}}^{\overline{M}_N} \tilde
Q^i(x,({1-p})/{p^2}h^2)}{\overline{M}_N-M_N^{\star}+1} \\
&&\qquad\plim \frac{ 4\alpha\pi(({1-p})/{p}h^2
)}{\varepsilon
^2} ,
\end{eqnarray*}
which can be arbitrarily small as $\alpha$ goes to $0$. Similarly, one
can bound the third term of (\ref{eq:majoBN})
and let $N$ go to infinity. Finally, letting $\alpha$ go to 0, we
obtain that $B_N /\sqrt{N}\plim0$. Thus,
$\Delta_{N,2} \plim0$. Finally, by Slutsky's lemma,
\[
\Delta_{N,3}:=(N/M_N^\star)^{-1/2}\frac{\sum_{i=1}^{M_N^\star}
\hat\xi
_i^0 h(\frx_i)}{\sqrt{M_N^\star}} \dlim
\calN\bigl(0,V_0[h-\pi(h)]/ \pi(p)\bigr) .
\]
The proof is thus complete.
\end{pf}

Note that the above analysis also provides us with a universal control
variate for Metropolis--Hastings algorithms.
Indeed, while Lemma \ref{lem:2uston} shows that
\[
\hat\xi_i = 1+\sum_{j=1}^\infty\prod_{\ell\le j} \{ 1 -
\alpha
(\frx_i,y_\ell) \}
\]
is an unbiased estimator of $1/p(\frx_i)$, a simple independent
estimator of $p(\frx_i)$ is provided
by $\alpha(\frx_i,y_0)$ when $y_0$ is an independent draw from
$q(Y|\frx
_i)$. While the variation in
this estimate may result in a negligible improvement in the control
variate estimation, it is nonetheless
available for free in all settings and should thus be exploited.

\section{Illustrations}\label{sec:ex}

We first consider a series of toy examples to assess the possible gains
brought about by the
essentially free Rao--Blackwellization. Our initial
example is a random walk Metro\-polis--Has\-tings algorithm with target
the $\mathcal{N}(0,1)$ distribution
and with proposal $q(y|x)=\varphi(x-y;\tau)$, a normal random walk with
scale $\tau$. The
acceptance probability is then the ratio of the targets, and Figure
\ref{fig:rbrw} illustrates
%
%
\begin{figure}

\includegraphics{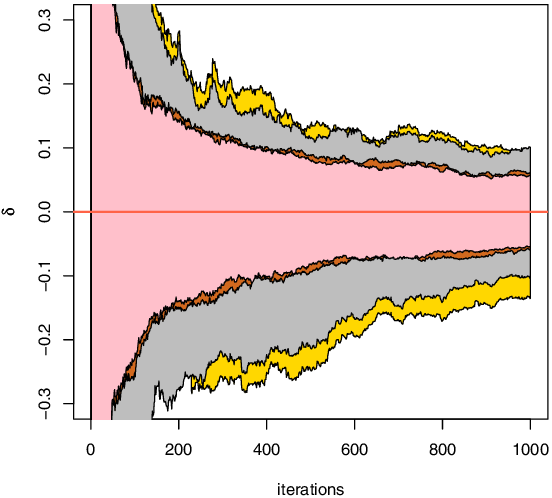}

\caption{Overlay of the variations of $250$ i.i.d. realizations of the
estimates $\delta$ \textup{(gold)}
and $\delta^\infty$ \textup{(grey)} of $\mathbb{E}[X]=0$ for $1000$
iterations, along with the $90\%$ interquantile
range for the estimates $\delta$ \textup{(brown)} and $\delta^\infty$
\textup{(pink)}, in the setting
of a random walk Gaussian proposal with scale $\tau=10$.}\label{fig:rbrw}
\end{figure}
the gain provided by the Rao--Blackwellization scheme by repeating the
simulation $250$
times and by representing
the $90\%$ range as well as the whole range of both estimators. The
gain provided by the Rao--Blackwellization is not huge
with respect to the overlap of both estimates, but one must consider
that the variability of the estimator
$\delta$ is due to two sources of randomness, one due to the $\frn_i$'s
and the other due
to the $\frx_i$'s. In addition, the gain forecasted by the above
developments is in terms of variance, not of tails, and
this gain is illustrated in Table \ref{tab:rbrw}. In this table, we
%
%
\begin{table}[b]
\tablewidth=250pt
\caption{Ratios of the empirical variances of the components of the estimators
$\delta^\infty$ and $\delta$
of $\mathbb{E}[h(X)]$ for $100$ MCMC iterations over $10^3$ replications,
in the setting of a random walk Gaussian proposal with scale $\tau$,
when started with a normal simulation}\label{tab:rbrw}
\begin{tabular*}{\tablewidth}{@{\extracolsep{\fill}}lllll@{}}
\hline
\multicolumn{1}{@{}l}{$\bolds{h(x)}$} & \multicolumn{1}{c}{$\bolds{x}$}
& \multicolumn{1}{c}{$\bolds{x^2}$} & \multicolumn{1}{c}{$\bolds{\mathbb{I}_{X>0}}$}
& \multicolumn{1}{c@{}}{$\bolds{p(x)}$}\\
\hline
$\tau=0.1$ &0.971 &0.953 &0.957 &0.207\\
$\tau=2$ &0.965 &0.942 &0.875 &0.861\\
$\tau=5$ &0.913 &0.982 &0.785 &0.826\\
$\tau=7$ &0.899 &0.982 &0.768 &0.820\\
\hline
\end{tabular*}
\end{table}
provide the ratio of the empirical variances of
the terms $\frn_i h(\frx_i)$ and $\hat\xi_i h(\frx_i)$ for several
functions $h$.
The minimal gains when $\tau=0.1$ are explained by the fact that the
acceptance probability is almost 1 with such a small scale, while the
higher rejection rate of $82\%$ when $\tau=7$ leads
to more improvement in the variances because of a higher variability in
the original $\frn_i$'s. Note that the last column
of Table \ref{tab:rbrw} estimates $\mathbb{E}[p(x)]$ via an additional
draw from $q(Y|\frx_i)$, as pointed out at the end of the
previous section. Table \ref{tab:rbt} gives an evaluation of the
%
\begin{table}
\tablewidth=280pt
\caption{Evaluations of the additional computing effort due to the use
of the Rao--Blackwell correction: median and mean numbers of additional
iterations, $80\%$ and $90\%$ quantiles for the additional iterations,
and ratio of the average R computing times obtained over $10^5$
simulations in the same setting as Table \protect\ref{tab:rbrw}}\label{tab:rbt}
\begin{tabular*}{\tablewidth}{@{\extracolsep{\fill}}lccccd{1.2}@{}}
\hline
& \textbf{Median} & \textbf{Mean} & $\bolds{q_{0.8}}$ & $\bolds{q_{0.9}}$
& \multicolumn{1}{c@{}}{\textbf{Time}}\\
\hline
$\tau=0.1$ &1.0 &6.49  &5.0 &11 &2.33\\
$\tau=2$ &0.0 &7.06 &4.3 &11 &6.5\\
$\tau=5$ &0.0 &9.02 &4.6 &13 &8.4\\
$\tau=7$ &0.0 &9.47 &4.8 &13 &3.5\\
\hline
\end{tabular*}
\end{table}
additional time required by the Rao--Blackwellization,
even though this should not be overinterpreted. As shown by both the
difference between the median and the mean additional times
and the variability of the increase in the R computing time, despite
the use of $10^5$ replications,
the occurrence of a few very lengthy runs accounts for the apparently
much higher computing times. Note that this difficulty with
very long runs can be completely bypassed when using a truncated
version $\delta^k$ instead of the unconstrained version $\delta
^\infty$.

%
\begin{figure}

\includegraphics{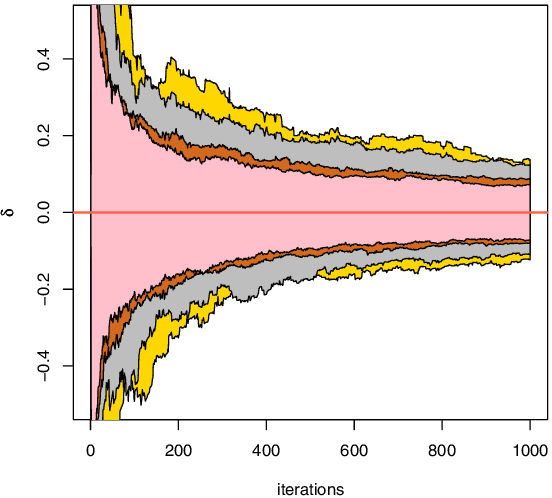}

\caption{Overlay of the variations of $250$ i.i.d. realizations of the
estimates $\delta$ \textup{(gold)}
and $\delta^\infty$ \textup{(grey)} of $\mathbb{E}[X]=0$ for $1000$
iterations, along with the $90\%$ interquantile
range for the estimates $\delta$ \textup{(brown)} and $\delta^\infty$
\textup{(pink)}, in the setting
of an independent Cauchy proposal with scale $0.25$.}\label{fig:rbind}
\end{figure}

Our second example is an independent Metropolis--Hastings algorithm
with target the
$\mathcal{N}(0,1)$ distribution
and with proposal a Cauchy $\mathcal{C}(0,0.25)$ distribution. The
outcome is quite similar, but producing
a slightly superior improvement, as shown in Figure \ref{fig:rbind}.
Table \ref{tab:rbind} also indicates much
more clearly that the gains in variance can be substantial. Once again,
Table \ref{tab:rbrwt} shows that the
computing time may vary quite widely due to a few outlying instances of
late acceptance.

Our third example is an independent Metropolis--Hastings algorithm with
target the
$\mathcal{E}\mathrm{xp}(\lambda)$ distribution
and with proposal the $\mathcal{E}\mathrm{xp}(\mu)$ distribution. In this case,
the probability functions $p(x)$
in (\ref{eq:defp}) and $r(x)$ in Proposition \ref{prop:var} can be
derived in closed form as
\[
p(x)=1-\frac{\lambda-\mu}{\lambda} e^{-\mu x}
\quad\mbox{and}\quad
r(x)=1-\frac{2(\lambda-\mu)}{2\lambda-\mu} e^{-\mu x} .
\]
This special case means that we can compare the variability of the
original Metropolis--Hastings estimator with its
Rao--Blackwellized version $\delta^\infty_M$, but also with the optimal
importance sampling version
shown in (\ref{eq:lgnInduite}). As illustrated by
Table~\ref{tab:expexp}, the gain brought about by the Rao--Blackwellization is
significant, even when compared with the reduction in variance of the
optimal importance sampling version.
Obviously, the most extreme case of $\mu=0.1$ shows that the ideal
importance sampling estimator
(\ref{eq:lgnInduite}) could bring considerable improvement, were it available.

%
\begin{table}[b]
\tablewidth=250pt
\caption{Ratios of the empirical variances of the components of the estimators
$\delta^\infty$
and $\delta$ of $\mathbb{E}[h(X)]$ for $100$ MCMC iterations over
$10^3$ replications,
in the setting of an independent Cauchy proposal with scale $\tau$
started with a normal simulation}\label{tab:rbind}
\begin{tabular*}{\tablewidth}{@{\extracolsep{\fill}}lllll@{}}
\hline
\multicolumn{1}{@{}l}{$\bolds{h(x)}$} & \multicolumn{1}{c}{$\bolds{x}$}
& \multicolumn{1}{c}{$\bolds{x^2}$} & \multicolumn{1}{c}{$\bolds{\mathbb{I}_{X>0}}$}
& \multicolumn{1}{c@{}}{$\bolds{p(x)}$}\\
\hline
$\tau=0.25$ &0.677 &0.630 &0.663 &0.599\\
$\tau=0.5$ &0.790 &0.773 &0.716 &0.603\\
$\tau=1$ &0.937 &0.945 &0.889 &0.835\\
$\tau=2$ &0.781 &0.771 &0.694 &0.591\\
\hline
\end{tabular*}
\end{table}


\begin{table}
\tablewidth=280pt
\caption{Evaluations of the additional computing effort due to the use
of the Rao--Blackwell correction: median and mean numbers of additional
iterations, $80\%$ and $90\%$ quantiles for the additional iterations,
and\break ratio of the average R computing times obtained over $10^5$\break
simulations in the same setting as Table
\protect\ref{tab:rbind}}\label{tab:rbrwt}
\begin{tabular*}{\tablewidth}{@{\extracolsep{\fill}}ld{1.2}cd{1.1}d{2.1}d{1.2}@{}}
\hline
& \multicolumn{1}{c}{\textbf{Median}} & \multicolumn{1}{c}{\textbf{Mean}}
& \multicolumn{1}{c}{$\bolds{q_{0.8}}$} & \multicolumn{1}{c}{$\bolds{q_{0.9}}$}
& \multicolumn{1}{c@{}}{\textbf{Time}}\\
\hline
$\tau=0.25$ &0.0  &8.85 &4.9 &13  &4.2\\
$\tau=0.50$ &0.0  &6.76 &4   &11  &2.25\\
$\tau=1.0$ &0.25 &6.15 &4   &10  &2.5\\
$\tau=2.0$ &0.20 &5.90 &3.5 &8.5 &4.5\\
\hline
\end{tabular*}
\end{table}

%
\begin{table}
\tablewidth=280pt
\caption{Ratios of the empirical variances of the components of the estimators
$\delta$
and $\delta^\infty$ of $\mathbb{E}[h(X)]$ for $100$ MCMC iterations
over $10^3$ replications,
in the setting of an independent exponential proposal with scale $\mu$
started with an
exponential $\mathcal{E}\mathrm{xp}(1)$ simulation from the target distribution;
the second row is the optimal
gain obtained by using $1/p(\frx_i)$ as importance weight, that is, the
importance sampling estimator
(\protect\ref{eq:lgnInduite})}\label{tab:expexp}
\begin{tabular*}{\tablewidth}{@{\extracolsep{\fill}}ld{1.4}d{1.4}cc@{}}
\hline
$\bolds{h(x)}$ & \multicolumn{1}{c}{$\bolds{x}$} &
\multicolumn{1}{c}{$\bolds{x^2}$} & $\bolds{\mathbb{I}_{X>1}}$ & $\bolds{p(x)}$\\
\hline
$\mu=0.9$ &0.933 &0.953 &0.939 &0.238\\
&0.787 &0.774 &0.859 &0.106\\
$\mu=0.5$ &0.722 &0.807 &0.759 &0.591\\
&0.291 &0.394 &0.418 &0.285\\
$\mu=0.3$ &0.671 &0.738 &0.705 &0.657\\
&0.131 &0.175 &0.263 &0.295 \\
$\mu=0.1$ &0.641 &0.700 &0.676 &0.703\\
&0.0561 &0.0837 &0.159 &0.289\\
\hline
\end{tabular*}
\end{table}

Our fourth and final toy example is a geometric $\mathcal{G}\mathrm{eo}(\beta)$
target associated with a one-step
random walk proposal:
\[
\pi(x)=\beta(1-\beta)^x\quad\mbox{and}\quad
2q(y|x)=\cases{\mathbb{I}_{|x-y|=1}, &\quad if $x>0$,\cr
\mathbb{I}_{|y|\le1}, &\quad if $x=0$.}
\]
For this problem,
\[
p(x) = 1-\beta/2\quad\mbox{and}\quad r(x)=1-\beta+\beta^2/2 .
\]
We can therefore compute the gain in variance
\[
\frac{p(x)-r(x)}{2p(x)-r(x)} \frac{2-p(x)}{p^2(x)} =2 \frac
{\beta
(1-\beta)(2+\beta)}{(2-\beta^2)(2-\beta)^2},
\]
which is optimal for $\beta=0.174$, leading to a gain of $0.578$, while
the relative gain in variance is
\[
\frac{p(x)-r(x)}{2p(x)-r(x)} \frac{2-p(x)}{1-p(x)} =\frac{(1-\beta
)(2+\beta)}{(2-\beta^2)},
\]
which is decreasing in $\beta$.

We now apply the Rao--Blackwellization to a probit modeling of the Pima
Indian diabetes study [\citet{venablesripley2002}].
The data set we consider covers a population of $332$ women who were at
least 21 years old,
of Pima Indian heritage and living near Phoenix, Arizona. These women
were tested for diabetes
according to World Health Organization (WHO) criteria. The data were
collected by the US National
Institute of Diabetes and Digestive and Kidney Diseases, and is
available with the basic \textsf{R} package.
The goal is to explain the diabetes variable in terms of the body mass index.
We use a standard representation of the diabetes binary variables $y_i$
as indicators $y_i = \mathbb{I}_{z_i>0}$ of latent variables $z_i$,
$z_i|\beta\sim\mathcal{N}(\bx_i^{\mathrm{T}}\beta,1)$,
associated with a standard regression model, that is, where the $\bx
_i$'s are $p$-dimensional covariates and
$\beta$ is the vector of regression coefficients. Given $\beta$, the
$y_i$'s are independent Bernoulli random variables
with $\mathbb{P}(y_i=1|\beta)=\Phi(\bx_i^{\mathrm{T}}\beta
)$, where $\Phi$ is the standard normal cumulative distribution function.
The choice of a prior distribution for the probit parameter $\beta$ is
open to debate [\citet{marinrobert2007}], but, for
the purposes of illustration, we opt for a flat prior. The Metropolis--Hastings
algorithm associated with the posterior is a simple two-dimensional
random walk proposal with a single scale $\tau$, due to the
normalization of the body mass index.
Simulations based on different scales $\tau$ show significant
improvements in the variance of the terms of
$\delta$ and $\delta_\infty$ by a factor of $2$. If we consider,
in addition, the possible improvement brought about by the control
variate indicated at the end of the previous section, the regression
coefficient can be obtained by a simple regression of $\hat\xi_i
h(\frx
_i)$ over $\hat\xi_i \alpha(\frx_i,y_0)$ and Table \ref{tab:pima}
%
%
\begin{table}
\tablewidth=250pt
\caption{Ratios of the empirical variances of the components of the\break estimators
$\delta$
and $\delta^\infty$ of $\mathbb{E}[h(\beta)]$ for $10^4$ MCMC iterations,
in the setting of a random walk proposal with scale $\tau$ started from
the MLE estimate of $\beta$
applied to the Pima Indian diabetes study; the second row for each
value of $\tau$ is the additional
improvement in the empirical variances resulting from using the control
variate}\label{tab:pima}
\begin{tabular*}{\tablewidth}{@{\extracolsep{\fill}}lccc@{}}
\hline
$\bolds{h(\beta)}$ & $\bolds{\beta_1}$ & $\bolds{\beta_2}$ & $\bolds{\mathbb{I}_{\beta_2>0.5}}$\\
\hline
$\tau=0.01$ &0.523 &0.516 &0.944 \\
&0.999 &0.999 &0.996 \\
$\tau=0.05$ &0.481 &0.518 &0.877 \\
&0.864 &0.888 &0.929 \\
$\tau=0.1$ &0.550 &0.555 &0.896 \\
&0.749 &0.748 &0.765 \\
$\tau=0.2$ &0.562 &0.568 &0.845 \\
&0.532 &0.527 &0.620 \\
$\tau=0.5$ &0.556 &0.565 &0.778 \\
&0.412 &0.433 &0.479 \\
\hline
\end{tabular*}
\end{table}
shows that this additional step brings about a significant improvement
over the Rao--Blackwellized version.

\section*{Acknowledgments}

Both authors are grateful to Elke Thonnes and Gareth Roberts for
organizing the MCMC workshop
in Warwick, March 2009, that led to the completion of this work.
Suggestions from the editorial team led to an improved presentation,
for which the authors are
most grateful.


%
\printaddresses


\begin{thebibliography}{99}

\bibitem[\protect\citeauthoryear{Casella and Robert}{1996}]{casellarobert1996}
\textsc{Casella, G.} and \textsc{Robert, C.} (1996).
{R}ao-{B}lackwellisation of sampling schemes.
\textit{Biometrika} \textbf{83} 81--94.
\MR{1399157}

\bibitem[\protect\citeauthoryear{Casella and Robert}{1998}]{casellarobert1998}
\textsc{Casella, G.} and \textsc{Robert, C.} (1998).
Post-processing accept-reject samples: Recycling and rescaling.
\textit{J. Comput. Graph. Statist.} \textbf{7} 139--157.
\MR{1649370}

\bibitem[\protect\citeauthoryear{Delmas and
Jourdain}{2009}]{delmasjourdain2009}
\textsc{Delmas, J. F.} and \textsc{Jourdain, B.} (2009).
Does waste recycling really improves the multi-proposal {M}etropolis
{H}astings algorithm? An analysis based on control variates.
\textit{J. Appl. Probab.} \textbf{46} 938--959.
\MR{2582699}

\bibitem[\protect\citeauthoryear{Douc and Moulines}{2008}]{doucmoulines2008}
\textsc{Douc, R.} and \textsc{Moulines, E.} (2008).
Limit theorems for weighted samples with applications to sequential
{M}onte {C}arlo methods.
\textit{Ann. Statist.} \textbf{36} 2344--2376.
\MR{2458190}

\bibitem[\protect\citeauthoryear{G{\aa}semyr}{2002}]{gasemyr2002}
\textsc{G{\aa}semyr, J.} (2002).
{M}arkov chain {M}onte {C}arlo algorithms with independent proposal
distribution and their relation to importance sampling and rejection
sampling.
Technical Report 2, Dept. Statistics, Univ. Oslo.

\bibitem[\protect\citeauthoryear{Latuszynski et
al.}{2010}]{latuszynskikosmidispaparoberts2010}
\textsc{Latuszynski, K.}, \textsc{Kosmidis, I.},
\textsc{Papaspiliopoulos, O.}
and \textsc{Roberts, G.} (2010).
Simulating event of unknown probabilities via reverse time
martingales.
\textit{Random Structures Algorithms}. To appear.

\bibitem[\protect\citeauthoryear{Malefaki and
Iliopoulos}{2008}]{malefakiiliopoulos2008}
\textsc{Malefaki, S.} and \textsc{Iliopoulos, G.} (2008).
On convergence of importance sampling and other properly weighted
samples to the target distribution.
\textit{J. Statist. Plann. Inference} \textbf{138} 1210--1225.
\MR{2381076}

\bibitem[\protect\citeauthoryear{Marin and Robert}{2007}]{marinrobert2007}
\textsc{Marin, J.-M.} and \textsc{Robert, C.} (2007).
\textit{Bayesian Core}.
Springer, New York.
\MR{2289769}

\bibitem[\protect\citeauthoryear{Perron}{1999}]{perron1999}
\textsc{Perron, F.} (1999).
Beyond accept--reject sampling.
\textit{Biometrika} \textbf{86} 803--813.
\MR{1741978}

\bibitem[\protect\citeauthoryear{Sahu and
Zhigljavsky}{1998}]{sahuzhigljavsky1998}
\textsc{Sahu, S.} and \textsc{Zhigljavsky, A.} (1998).
Adaptation for self regenerative {MCMC}.
Technical report, Univ. Wales, Cardiff.

\bibitem[\protect\citeauthoryear{Sahu and
Zhigljavsky}{2003}]{sahuzhigljavsky2003}
\textsc{Sahu, S.} and \textsc{Zhigljavsky, A.} (2003).
Self regenerative {M}arkov chain {M}onte {C}arlo with adaptation.
\textit{Bernoulli} \textbf{9} 395--422.
\MR{1997490}

\bibitem[\protect\citeauthoryear{Venables and
Ripley}{2002}]{venablesripley2002}
\textsc{Venables, W.} and \textsc{Ripley, B.} (2002).
\textit{Modern Applied Statistics with {S-PLUS}},
4th ed. Springer, New York.
\MR{1337030}

\end{thebibliography}
\end{document}